\documentclass[conference]{IEEEtran}
\IEEEoverridecommandlockouts

\usepackage{amsthm, amssymb}

\def\b{\mathbb}

\def\c{\mathcal}
\def\sinr{\textrm{SINR}}
\def\snr{\textrm{SNR}}

\newtheoremstyle{slplain}
  {3pt}
  {3pt}
  {\slshape}
  {}
  {\bfseries}
  {.}%
  { }
  {}

\theoremstyle{slplain}
\newtheorem{thm}{Theorem}
\newtheorem{cor}{Corollary}

\usepackage{cite}
\usepackage{amsfonts}
\usepackage{multicol,multienum}
\usepackage{subfigure}

\ifCLASSINFOpdf
  \usepackage[pdftex]{graphicx}
  \graphicspath{{../pdf/}{../jpeg/}}
  \DeclareGraphicsExtensions{.pdf,.jpeg,.png}
\else
  \usepackage{float}
  \usepackage[dvips]{graphicx}
  \graphicspath{{../eps/}}
  \DeclareGraphicsExtensions{.eps}
\fi

\usepackage[cmex10]{amsmath}
\usepackage{algorithm}
\usepackage{algorithmic}
\usepackage{array}
\usepackage{url}
\usepackage{setspace}

\begin{document}


\title{Dynamic Spectrum Refarming of GSM Spectrum for LTE Small Cells}

\author{ Xingqin Lin and Harish Viswanathan
\thanks{Xingqin Lin is with Department of Electrical \& Computer Engineering at the University of Texas, Austin, TX. Harish Viswanathan is with Bell Labs, Alcatel-Lucent, Murray Hill, NJ. This work was performed when Xingqin Lin was a summer intern at Bell Labs. Email: xlin@utexas.edu, harish.viswanathan@alcatel-lucent.com. Date revised: \today}
}

\maketitle

\begin{abstract}
In this paper we propose a novel solution called dynamic spectrum refarming (DSR) for deploying LTE small cells using the same spectrum as existing GSM networks. The basic idea of DSR is that LTE small cells are deployed in the GSM spectrum but suppress transmission of all signals including the reference signals in some specific physical resource blocks corresponding to a portion of the GSM carriers to ensure full GSM coverage. Our study shows that the proposed solution can provide LTE mobile terminals with  high speed data services when they are in the coverage of the LTE small cells while minimally affecting the service provided to GSM terminals located within the LTE small cell coverage area. Thus, the proposal allows the normal operation of the existing GSM networks even with LTE small cells deployed in that spectrum.  Though the focus of this paper is about GSM spectrum refarming, an analogous approach can be applied to reuse code division multiple access (CDMA) spectrum for LTE small cells.
\end{abstract}

\IEEEpeerreviewmaketitle


\section{Introduction}

Our world is witnessing an inexorable and exponential wireless traffic growth \cite{cisco2011cisco}. Small cells, involving isolated or clustered deployment of low power base stations, is emerging as an effective approach to address the high capacity demand in traffic hot spots.
Long Term Evolution (LTE) is the key standard for high speed data services in cellular networks \cite{website:3gppLTE}. The deployment of LTE small cells provides a fast and inexpensive way to meet the intense user demand for mobile traffic \cite{Damnjanovic2011survey, Ghosh2012Heterogeneous}. Deploying LTE small cells in the same spectrum requires careful handling of the interference between the macrocells and small cells while deploying it in separate spectrum requires additional dedicated spectrum for small cells. Refarming of the spectrum used by older technologies such as Global System for Mobile (GSM) \cite{mehrotra1997gsm}, the demand for which is diminishing, offers the possibility to make additional spectrum available for use by LTE.

In this paper we propose to deploy LTE small cells in areas of high demand in existing GSM spectrum. Deploying LTE small cells in GSM spectrum allows reusing highly valuable spectrum inefficiently used by an old technology for a highly spectrally efficient solution while still maintaining the old technology for legacy devices. However, a direct deployment of LTE small cells in existing GSM networks may cause degradation of legacy communications. In particular, since the LTE small cells use the same spectrum as the GSM macro cells, the legacy GSM devices would experience severe LTE interference when they are located in the coverage of the small cells.

Frequency partitioning between GSM and LTE is one potential approach to deal with the cross-tier interference mentioned above. Specifically, we may configure the LTE small cells such that the LTE transmission bandwidth only occupies part of the available spectrum. The GSM devices located within the coverage of the small cells can avoid the severe LTE interference by using the channels not used by the small cells. Though conceptually simple, this may involve complex frequency planning. More importantly, this approach disallows the operators to refarm completely GSM spectrum for LTE deployment and is also not efficient (see \cite{lin2013dynamic} for more detailed discussion).

Given the limitation of frequency partitioning, we propose a novel approach that we term dynamic spectrum refarming (DSR) \cite{lin2013dynamic} for the deployment of LTE small cells in existing GSM networks. Specifically, we allow LTE small cells to be configured such that they  completely refarm the GSM spectrum, subject to allowed LTE system bandwidths. At the same time, some specific physical resource blocks (PRB) corresponding to a portion of the GSM carriers will be reserved for GSM transmission, i.e., LTE Evolved Node B (eNodeB) will not schedule those reserved PRBs for any User Equipment (UE) and accordingly suppress the reference signals.  With this approach, the GSM devices located within the coverage of the small cells can be protected from the LTE interference since these devices can now choose GSM channels in the portion of the spectrum where LTE small cell signal is suppressed. So DSR prevents the GSM coverage holes induced by LTE small cells, while it allows a complete refarming of GSM spectrum for LTE deployment. UEs can be covered by GSM networks while experiencing high speed data services provided by LTE connections when they are in the coverage of small cells. DSR thus provides a flexible approach for the co-existence of LTE small cells with GSM networks.

However, there are several technical issues to be addressed. First and foremost, will LTE UEs with no or only limited modifications continue to operate as usual? This is not a priori clear since DSR modifies the channel structure of LTE. Secondly, what is the impact of LTE small cells on the performance, in particular outage, of those GSM users in the boundary of the small cells? Last but not least, how good is the performance e.g. throughput of the LTE small cells? We address these critical issues in this paper. While the paper is focused on LTE small cells overlay over GSM, the concept and analysis are generally applicable for deployment of any wideband OFDM system overlay over a narrower band system.


\section{LTE Small Cells in GSM Spectrum}
\label{sec:lte}

In this section, we provide detailed description of the system design and argue the feasibility of DSR, i.e., LTE with suppression of some PRBs which is subsequently referred to as LTE with puncturing. While small cells use LTE with puncturing, two options are available for the operation of the GSM macro cells. The first option is to maintain the same amount of spectrum for GSM but to adopt intelligent channel assignment in GSM, i.e., allocate channels corresponding to the suppressed PRBs to those GSM users located in the coverage of small cells  in that sector. The other option is to reduce the amount of spectrum for GSM by scheduling all the GSM signals only on the LTE PRBs suppressed in that sector. The second option is conceptually simpler and has the advantage that GSM signals are (nearly) free of LTE interference but with a cost of reduced GSM capacity. In the sequel, we focus on the first option noting that the concept and analysis is also applicable to the other option. Also, to make the description concrete, we focus on the downlink system. Uplink system can be treated along similar lines.

\subsection{DSR: From Theory to Practice}

DSR solution reserves certain PRBs within the LTE transmission bandwidth for GSM signals and does not transmit any LTE signal on those sub-carriers. This reservation scheme changes the channel structure of LTE. To ensure normal operation of LTE, we need to carefully select the positions of those reserved PRBs out of the LTE channel. We take the $10$MHz LTE channel as an example and discuss the potential impact of LTE puncturing on the broadcast and synchronization channels, reference symbols, and other control signaling channels including Physical Control Format Indicator Channel (PCFICH), Physical Downlink Control Channel (PDCCH) and Physical Hybrid ARQ Indicator Channel (PHICH) \cite{3gppMobility}.

As argued in \cite{lin2013dynamic},  we can avoid choosing the PRBs used by the LTE broadcast and synchronization channels, and puncturing reference symbols located within the reserved PRBs is not an issue. However, the positions of PCFICH and PHICH are not fixed and vary with the physical cell ID \cite{3gppMobility}. They would spread out and occupy all the PRBs  if all the cell IDs were used and we would not be able to reserve PRBs without affecting them. To overcome this issue, we propose to use only a small subset of the cell IDs. Then PCFICH and PHICH will only occupy certain part of the LTE transmission bandwidth, which can be avoided by the reserved PRBs. Meanwhile, since the small cells are often clustered in the areas of high population density and different clusters are often geographically separated from each other, a small number of cell IDs seems to be enough for practical use in each cluster and the same set of cell IDs can be reused from one cluster to another.

The real challenge comes from the PDCCH which in any case occupies all the PRBs. This implies that a certain number of resource element groups have to be wiped out with LTE puncturing. Fortunately, LTE multiplexes and interleaves several PDCCHs within one LTE subframe. This helps to spread the impact of puncturing over all the PDCCHs and each PDCCH only needs to tolerate a certain level of the errors. Moreover, LTE allows the system to increase the aggregation level of the control channel elements, which can make PDCCHs more robust against errors.  Note that given limited resource of the control channel elements, increasing its aggregation level reduces the number of PDCCHs that can be simultaneously used. However, we believe that this is not a critical issue in small cells which have limited coverage and need to serve only a small number of UEs. The link level simulation for studying the impact of puncturing on PDCCH can be found in \cite{lin2013dynamic}, which shows that the PDCCH block error rate (BLER) increases due to the puncturing. However, the loss in BLER can be recovered by increasing the aggregation level of the control channel elements by one.


To sum up, though DSR inevitably affects the channel structure of LTE, it appears feasible with a careful design.

\subsection{Deployment of LTE Small Cells}

In the previous subsection, we argued the feasibility of LTE operation despite puncturing a few PRBs. However, to ensure successful co-existence of the small cells with the GSM macro cells, we still need to address the following key issues.

\subsubsection{How to set the transmit power of the small cells} Though GSM users located within the coverage of small cells are free of LTE interference by scheduling their signals on the punctured PRBs, those GSM users located close to the small cells may still suffer from severe interference from the small cells. So we need to carefully set the transmit power of the small cells to avoid excessive interference to these nearby GSM users. In this paper we propose to use the concept of guard region around each small cell (cf. Fig. \ref{fig:model2}).  The guard region encompasses the coverage area of the small cell and extends further out. The macro cell schedules GSM users located within the guard region of the small cell on the reserved PRBs. The transmit power of the small cell is then configured such that the average received Signal to Interference plus Noise Ratio (SINR) degradation of those GSM users located at the border of the guard region should be less than e.g. 1 dB. 

Obviously, the size of the guard region is closely related to the number of reserved PRBs and the traffic distribution over the space. For example, assuming a uniform spatial traffic distribution, the number of reserved PRBs should be proportional to the size of the guard region: The larger the guard region, the more PRBs we need to reserve. Given the size of the guard region, the transmit power of the small cells can be determined, which in turn determines the coverage area of the small cells.

\subsubsection{How to puncture the GSM spectrum in each sector}
Note that GSM adopts frequency reuse, implying that LTE small cells located in different GSM macro sectors have to puncture different parts of the transmission bandwidth. To make the discussion specific, we assume a particular type of frequency planning for the macro cells using 10MHz spectrum, as shown in Fig. \ref{fig:gprsBW}. There are 48 carriers with 200KHz bandwidth each while the remaining 400KHz spectrum is used as guard band. The example frequency planning shown assumes 1/3 frequency reuse for the traffic carrier while 3/9 frequency reuse is assumed for the control carrier. Each GSM sector uses $13$ traffic carriers as well as $1$ control carrier.  LTE small cells use the whole $10$MHz spectrum as its transmission bandwidth. Suppose the guard region of the small cell in each sector is about $30\%$ as large as the coverage of the associated sector, the small cell in that sector can reserve 1MHz of the spectrum allocated to that sector, as shown in Fig. \ref{fig:gprsBW}. Note that the puncturing  in Fig. \ref{fig:gprsBW} is just an illustrative example. The real puncturing should also satisfy the general rules described in the previous section to make the LTE system work.

\begin{figure}
\centering
\includegraphics[width=8.5cm]{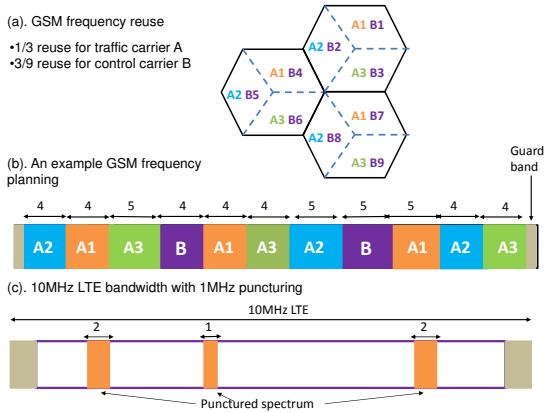}
\caption{An example of frequency planning for the GSM networks and LTE with puncturing.}
\label{fig:gprsBW}
\end{figure}

\section{System Model}
\label{sec:sys}

In this section, we describe a specific system model to study the throughput of LTE small cells and the outage performance of GSM networks. We consider a two-tier cellular network, where  LTE small cells are overlaid within GSM macro cells, as shown in Fig. \ref{fig:model2}. We adopt the frequency planning  shown in Fig. \ref{fig:gprsBW} for the GSM macro cells.  The side length of the hexagonal cell of the macro GSM BS is denoted by $R_m$. The intended GSM sector is overlaid by one LTE small cell, which is at the position $(D,\theta)$ w.r.t. its associated GSM BS. The coverage radius and guard region radius of the LTE small cell are denoted by $R_c$ and $R_s$, respectively. The macro BSs are equipped with directional antennas while LTE small BSs use omni-directional antennas to minimize complexity. The directional antennas are assumed to have the radiation pattern $G(\phi)$.

\begin{figure}
\centering
\includegraphics[width=8cm]{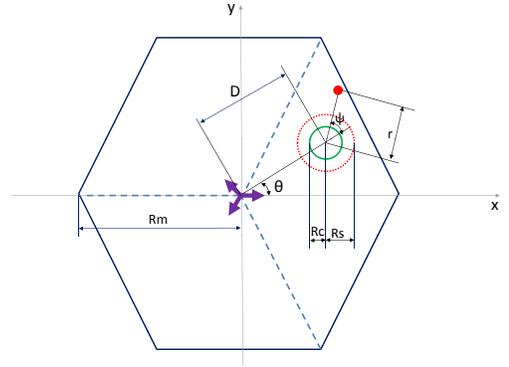}
\caption{System model: Green ball denotes the coverage area of the LTE small cell while red ball denotes the guard region.}
\label{fig:model2}
\end{figure}

We denote by $P_g$ and $P_l$ the transmit power of the GSM BSs and the LTE BSs, respectively. As described in Section \ref{sec:lte}, $P_l$ is configured such that the performance degradation of those GSM users located at the distance $R_s$ w.r.t. the LTE BS should be less than 1 dB. As for the GSM users located nearer than $R_s$ around the LTE BS, their signals are scheduled on the GSM channels corresponding to punctured PRBs.

The above description focuses on the hexagonal model for the macro cells. For analytical tractability, we further adopt a spatial random process based model. In particular, we assume the GSM BSs positions follow a Poisson point process (PPP) $\Phi$ with density $\lambda = \frac{2}{3\sqrt{3}R_m^2}$. This Poisson assumption on the BS positions has been used in e.g. \cite{brown2000cellular, andrews2010tractable, lin2012towards}.

\section{Analytical Results}
\label{sec:ana}

\subsection{LTE Throughput Calculation}

In this subsection, we study the achievable rates of LTE UEs in the small cells. To this end, we first notice that, due to the sectorization, a particular PRB  depending on its position within the spectrum experiences one of the 3 different types of GSM interference. We group the PRBs into three bands based on the experienced GSM interference. The bandwidth of these bands are denoted by $B_1 = 2$MHz, $B_2=B_3 = 3$MHz. That is, we puncture a $1$MHz portion in the LTE transmission bandwidth and band 1 corresponds to the GSM channels used in the sector on which the LTE small cell overlays, while bands 2 and 3 correspond to the GSM channels used by the other two neighboring sectors.

Then the band $i$'s SINR of the UE with radio link length $r$ is given by
\begin{align}
\sinr_i (r) = \frac{ S_l(r) }{ I_i + \sigma^2 }.
\end{align}
Here $\sigma^2$ is the noise variance, $S_l$ is the desired LTE signal power given by
\begin{align}
S_l(r) = \hat{P}_l H_l l(r) = \hat{P}_l H_l r^{-\alpha},
\end{align}
where $\hat{P}_l = P_l L_0$ with $L_0$ being the constant that gives the path loss when $r = 1$, $l(r) = r^{-\alpha}$ denotes the deterministic path loss law, and $H_l$ is a random variable capturing the fading effect of the LTE radio link. For simplicity, we ignore the shadowing effect and consider Rayleigh fading only, i.e., $H_l \sim \textrm{Exp}(1)$. This is reasonable since the small cell has small coverage and thus the shadowing effect is not that serious. Assuming w.l.o.g. the LTE UE is located at the origin, the GSM interference $I_i$ of type $i$  is given by
\begin{align}
I_i = \sum_{ x \in \Phi \slash B(o,R_i) } \hat{P}_g F_x  l( \| x \| ) G( \theta_i - \angle x ), i=1,2,3,
\end{align}
where $\hat{P}_g = P_g L_0$, $F_x$ is the random variable capturing the fading effect (which may include the effect of both shadowing and Rayleigh fading) from macro BS $x$ to the LTE UE,\footnote{We use $x$ to denote both the position of an arbitrary GSM BS and its index. The meaning should be clear from the context.} and $\theta_i - \angle x$ is the antenna azimuth with  $\theta_i = \frac{2\pi}{3} \cdot (i-1)$ being the antenna angular offset of the sector $i$ in the cell $x$, and $B(o,R_i)$, the ball of radius $R_i$ cetered at the origin, denotes the guard zone in which no macro BSs can occupy. This guard zone is used to capture the fact that the closest GSM interferer is at least at some distance $R_i$ away from the UE. 
Under our model, we choose $R_i$ as follows:
\begin{align}
R_1 = D - R_c, R_2 = R_3 = \sqrt{ ( \frac{3}{2}R_m - D + R_c )^2 + \frac{3}{4}R_m^2 }. \notag
\end{align}

Note in characterizing the UE's SINR, we assume GSM signals are the sole source of interference contributing to $I_i$ and the inter-cell interference of LTE small cells is ignored. This is reasonable because the overlaid LTE small cells operate at very low power and are typically well geographically separated from each other. In addition, we ignore the leakage of GSM transmission in the suppressed LTE PRBs into the desired LTE signal due to both adjacent channel leakage and in-channel selectivity, since the leakage is small compared to the interference captured in $I_i$ \cite{lin2013novel,lin2013dynamic}. 

After characterizing the SINRs, the UE's rate can be computed as
\begin{align}
\tau (r) &\triangleq 
&= \frac{1}{N_{ue}} \sum_{i=1}^3 \b E \left[ B_i \log ( 1 + \frac{1}{\eta} \sinr_i(r) ) \right],
\end{align}
where $N_{ue}$ denotes the expected number of UEs in the small cell and $\eta$ denotes the Shannon gap. To be specific, we assume UEs have density $\lambda^{(u)}$ and are uniformly and independently distributed within the coverage of the small cell. Here we assume a round robin scheduling scheme for the LTE small cell and each UE gets a fraction $\frac{1}{N_{ue}}$ of the total transmission resource. However, the proportional fair scheduling is more popular in current cellular networks and will be used in the simulation in this paper too.

We next give the explicit analytical expression for $\tau (r)$ in Theorem \ref{thm:1}.
\begin{thm}
The rate $\tau (r)$ of the LTE UE with link length $r$ is given by
\begin{align}
&\tau (r) = \frac{1}{\lambda^{(u)} \pi R_c^2} \cdot \notag \\
& \sum_{i=1}^3 B_i \int_0^\infty  e^{ -  \snr_l^{-1} r^\alpha \eta(e^t-1) }  \c L_{I_i} \left( \eta(e^t-1)r^\alpha \hat{P}^{-1}_l  \right)   dt,
\end{align}
where $\snr_l = \hat{P}_l/\sigma^2$ and
\begin{align}
\c L_{I_i} (s) = e^{ \lambda \pi R_i^2 +  \lambda \int_{0}^{2\pi} \frac{ s^{ \frac{2}{\alpha} } }{\alpha} \b E_{Z} \left[   Z^{ \frac{2}{\alpha} }(\phi) \cdot  \gamma ( -\frac{2}{\alpha}, sZ(\phi) R_i^{ - \alpha } ) \right] \ d\phi   }
\label{eq:thm:1}
\end{align}
with $Z(\phi) = \hat{P}_g   F G( \phi )$ and $\gamma(s,x) = \int_{0}^{x} t^{s-1} e^{-t} dt $ is the lower incomplete Gamma function.
\label{thm:1}
\end{thm}
\begin{proof}
See Appendix \ref{proof:thm1}.
\end{proof}

\subsection{GSM Outage Calculation}

In this subsection, we study the impact of the overlaid LTE small cells on the performance of the GSM networks, particularly outage probability. With 1/3 frequency reuse, power control is needed for traffic channels to maintain acceptable outage performance. Since our model does not incorporate power control, we consider the outage probability of the 3/9-reuse control channels without power control. In addition, when evaluating the SINR of a GSM user, we assume the cross-tier interference only comes from the nearest LTE small cell because it is the dominant one and other LTE small cells are relatively far away and operate at very low power and thus their interference can be ignored.

Assume the GSM user is located at $(r,\psi)$ w.r.t. the BS in the small cell, as shown in Fig. \ref{fig:model2}. Following the same line of the analysis as in the previous subsection, the SINR  of the GSM user is given by
\begin{align}
\sinr_g (r,\psi) =  \frac{ S_g (r,\psi) }{ I_{0} + I_{l} + \sigma^2 },
\label{eq:gsinr}
\end{align}
where $S_g (r,\psi) = \hat{P}_g F_g G(\beta(r,\psi)) d^{-\alpha}(r,\psi)$ is the desired GSM signal, $
I_l = \hat{P}_l F_l r^{-\alpha}
$ is the LTE interference, and $I_0$ is the co-channel inter-cell GSM interference. For simplicity, we ignore the shadowing effect of the desired GSM signal, i.e., $F_g \sim Exp(1)$.


We next define the outage probability of the GSM user located at $(r,\psi)$  as
\begin{align}
P_{out} (r,\psi) = P( \sinr_g (r,\psi) < T ),
\end{align}
where $T$ is the SINR threshold for reliable GSM transmissions.  The explicit analytical expression for $P_{out} (r,\psi)$ is given in Theorem \ref{thm:2}.

\begin{thm}
The outage probability of the GSM user located at $(r,\psi)$ w.r.t. the small cell at $(D,\theta)$ is given by
\begin{align}
P_{out} ( r, \psi ) = 1 - e^{ - s^\star \sigma^2 } \cdot \c L_{I_0} (s^\star) \cdot  \c L_{I_l} (s^\star).
\end{align}
Here  $s^\star = \frac{T d^{\alpha} (r,\psi)}{\hat{P}_g G(\beta(r,\psi)) }$, and
\begin{align}
d(r,\psi) = \sqrt{ D^2 + r^2 - 2rD \cos ( \pi - \psi ) },
\end{align}
and
\[ \beta(r,\psi) = \left \{ \begin{array}{ll}
\theta + \tilde{\beta} (r,\psi) & \mbox{if $0\leq \psi < \pi $}; \\
\theta - \tilde{\beta} (r,\psi) & \mbox{if $\pi \leq \psi \leq 2\pi$}; \end{array} \right. \]
where
\begin{align}
\tilde{\beta} (r,\psi) = \arccos \frac{ D^2 + d^2(r,\psi) - r^2 }{ 2 D d(r,\psi) }.
\end{align}
$\c L_{I_0} (s)$ is given in (\ref{eq:thm:1}) with $\lambda$ replaced by $\lambda/3$ and
\begin{align}
R_0 = \sqrt{ 9R_m^2 + D^2 - 6 D R_m \cos( \frac{2\pi}{3} - | \theta | )  },
\end{align}
and $\c L_{I_l} (s)$ is the Laplace transform of $I_l$. In particular, if $F_l \sim Exp(1)$,
\begin{align}
\c L_{I_l} (s) =   \frac{1}{1+ s  \hat{P}_l r^{-\alpha}}.
\end{align}
\label{thm:2}
\end{thm}
\begin{proof}
See Appendix \ref{proof:thm2}.
\end{proof}

With Theorem \ref{thm:2}, we can readily evaluate the outage performance of the original GSM network, where small cells are not deployed, as in the following corollary.
\begin{cor}
The outage probability of the GSM user located at $(r,\psi)$ w.r.t. the position $(D,\theta)$ but without the overlaid small cells is given by
\begin{align}
\hat{P}_{out} ( r, \psi ) = 1 - e^{ - s^\star \sigma^2 } \cdot \c L_{I_0} (s^\star),
\end{align}
where the notations are the same as in Theorem \ref{thm:2}.
\label{cor:1}
\end{cor}

With Theorem \ref{thm:2} and Corollary \ref{cor:1}, the average outage probabilities of the GSM users located within the ball with radius $R_s + \Delta R_s$ centered at $(D,\theta)$ in the three scenarios: No LTE overlay, Direct LTE overlay (without DSR), LTE overlay with DSR, are respectively given by
\begin{align}
P^{\textrm{w/o LTE}}_{out} = &G \int_{0}^{2\pi} \int_{0}^{R_s + \Delta R_s}  \hat{P}_{out} ( r, \psi )  r dr d \psi , \\
P^{\textrm{w/o DSR}}_{out} = &G \int_{0}^{2\pi} \int_{0}^{R_s + \Delta R_s}  P_{out} ( r, \psi )  r dr d \psi  ,\\
P^{\textrm{DSR}}_{out} = &G \int_{0}^{2\pi}   \int_{0}^{R_s}  \hat{P}_{out} ( r, \psi )  r dr  d \psi \notag \\ &+ G \int_{0}^{2\pi} \int_{R_s}^{R_s + \Delta R_s}  P_{out} ( r, \psi )  r dr  d \psi.
\end{align}
where $G^{-1} =   \pi (R_s + \Delta R_s)^2$. We will numerically compare these outage probabilities in Section \ref{sec:num}.

\section{Numerical and Simulation Results}
\label{sec:num}

In this section we provide some numerical and simulation results. The network layout used in the simulation consists of regular hexagonal cells. The center cell is treated as the typical one, while the remaining cells act as interfering cells. The radiation pattern used for the directional antennas of GSM BSs is based on the one suggested by 3GPP \cite{3gppMimoSim}:
\begin{align}
G(\phi) =  - \min \left( 12 \left(\frac{\phi}{70}\right)^2, 20  \right) \textrm{ dB }, -180 \leq \phi \leq 180. \notag
\label{eq:2}
\end{align}
The channel model used accounts for path loss and Rayleigh fading. For fair comparison with the analytical results, we ignore the shadowing. Besides, proportional fair scheduling is adopted in the simulation. The main simulation and/or numerical parameters used are summarized in Table \ref{tab:sys:para} unless otherwise specified.

\begin{table}
\centering
\begin{tabular}{|l||r|} \hline
$\#$ of hexagonal cells & $6 \times 6$  \\ \hline
Length of cell edge $R_m$  & $1000$m  \\ \hline
Coverage of the small cell $R_c$   & $50$m  \\ \hline
Location of the small cell $(D,\theta)$   & $(0.5R_m, 0)$  \\ \hline
Guard region of the small cell $R_s$   & $\pi R_s^2 =\frac{1}{4} \times \frac{\sqrt{3}}{2} R_m^2 $  \\ \hline \hline
GSM BS power & $40$W \\ \hline
Noise PSD & $-174$dBm \\ \hline
Wavelength $\lambda$ & 0.375m \\ \hline
Path loss exponent $\alpha$ & 3 \\ \hline
LTE transmission bandwidth & $10$MHz \\ \hline
$(B_1,B_2,B_3)$ & $(2,3,3)$MHz \\ \hline \hline
GSM carrier bandwidth & $200$KHz \\ \hline
GSM frequency reuse factor & $1/3$ \\ \hline
GSM antenna radiation $G(\phi)$ &  In Eq. (\ref{eq:2}) \\ \hline \hline
Gap $\eta$ & $3$dB \\ \hline
\end{tabular}
\caption{Simulation and/or Numerical Parameters}
\label{tab:sys:para}
\end{table}

\subsection{LTE Throughput}

We first show the mean LTE throughput of the small cell versus the position $(D,\theta)$ of the small cell in Fig. \ref{fig:meanR}. The results show that the LTE small cell can provide quite high throughput to UEs within its coverage. This demonstrates the potential of LTE small cells to provide high speed data services in existing GSM networks.
\begin{figure}
\centering
\includegraphics[width=8cm]{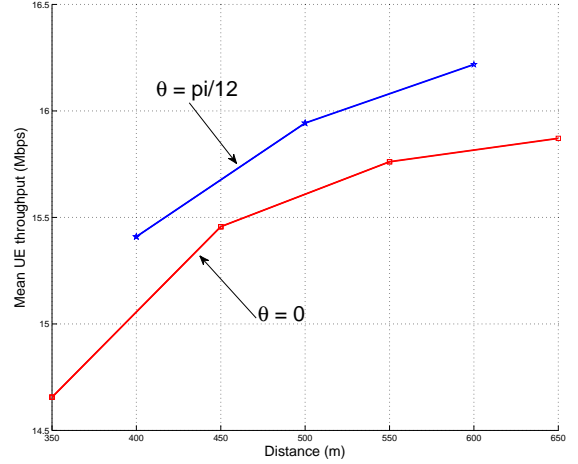}
\caption{Mean LTE throughput of the small cell: Here we only simulate $(D,\theta)$ such that the guard region of the small cell is fully contained in the corresponding GSM macro cell (cf. Fig. \ref{fig:model2}).}
\label{fig:meanR}
\end{figure}

We next compare the UEs' rates obtained by simulation to those computed from analytical result. The results are shown in Table \ref{tab:mrate}. We can see that the analytical approach underestimates the achievable rate by no more than $15\%$. This gap mainly arises from two factors. The first one is the difference between the hexagonal model used in the simulation and the PPP model used in the analysis.  It is known that the irregular Poisson-Voronoi tessellation leads to conservative results while hexagonal tessellation results in optimistic results \cite{andrews2010tractable}.  The second factor is the difference between the proportional fair scheduling used in the simulation and the round robin scheduling used in the analysis. Proportional fair scheduling exploits the channel conditions while round robin scheduling is channel-independent. Thus, it is natural to expect higher rate by using proportional fair scheduler.
\begin{table}
\centering
    \begin{tabular}{ | l | r | r | c | r | }
    \hline
     UE Position    & Simulation & Analysis & Gap & Gap/Sim \\ \hline
    (-22.1,	4.6)    & 17.53 & 14.91  & 2.62 & 0.15 \\ \hline
    (17.8,	25.7)   & 15.25 & 12.97   & 2.28 & 0.15 \\ \hline
    (-7.8,	41.5)   & 12.32 & 11.20   & 1.12 & 0.09 \\ \hline
    (30.0,	-35.8)  & 12.10 & 10.61   & 1.49 & 0.12 \\ \hline
    \end{tabular}
        \caption{UE's rate in the LTE small cell positioned at $D=350m,\theta=0$: Simulation versus analysis}
       \label{tab:mrate}
\end{table}

Though being a bit conservative, the analytical result, however, is very useful in exploring how the rate of the LTE small cell depends on other system  parameters, which may be too time-consuming to explore via simulation. For example, we numerically compute the LTE UEs' rates versus the size of the GSM macro cell and show the results in Fig. \ref{fig:macroSize}. The positions of those UEs are given in Table \ref{tab:mrate}. We observe an almost linear growth in the UEs' rates as $R_m$ increases. This growth arises from the reduced macro interference as $R_m$ increases.

 \begin{figure}
 \centering
 \includegraphics[width=8cm]{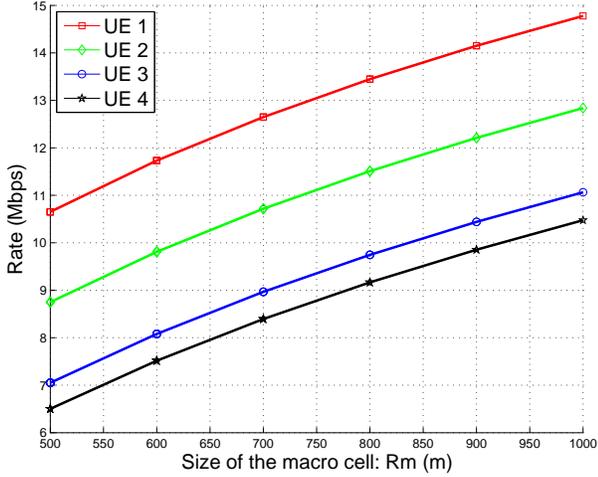}
 \caption{LTE UEs' rates versus the size of GSM macro cell}
 \label{fig:macroSize}
 \end{figure}

\subsection{GSM Outage}

We first compare the GSM outage probabilities obtained by simulation to those computed from analytical result. The results are shown in Fig. \ref{fig:outSvsA} and \ref{fig:outSvsALTE}. We can see that the gap between the analytic result and the simulation result for each user position is about 4dB at the $10\%$ outage performance. Thus, the analytical approach provides a relatively good estimate on the outage performance of GSM compared to the simulation result obtained from hexagonal model.     Moreover, the analytical results allow us to easily compare the average outage probabilities of the GSM users located within the ball with radius $R_s + \Delta R_s$ centered at $(D,\theta)$ in the three scenarios: No LTE overlay, Direct LTE overlay (without DSR), LTE overlay with DSR. The numerical results are plotted in Fig. \ref{fig:out_350_0}, which shows that a direct deployment of LTE small cells within the GSM networks severely degrades the GSM outage performance. In contrast, our proposed solution the LTE overlay with DSR only lightly degrades the GSM outage performance. Thus, our proposed scheme provides a good solution to the co-existence of LTE small cells and GSM macro cells.

 \begin{figure}
 \centering
 \includegraphics[width=8cm]{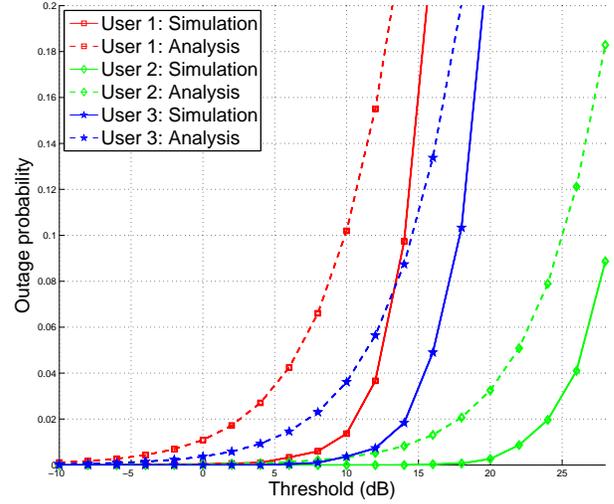}
 \caption{Outage performance of GSM without LTE overlay by simulation and analysis: The positions of the 3 GSM users w.r.t. their associated GSM BS are $(D+R_s, \theta), (\sqrt{D^2-R^2_s}, \theta+\Delta \theta ), (D-R_s, \theta)$ where $\Delta \theta= \arccos \frac{R_s}{D}$.}
 \label{fig:outSvsA}
 \end{figure}

  \begin{figure}
 \centering
 \includegraphics[width=8cm]{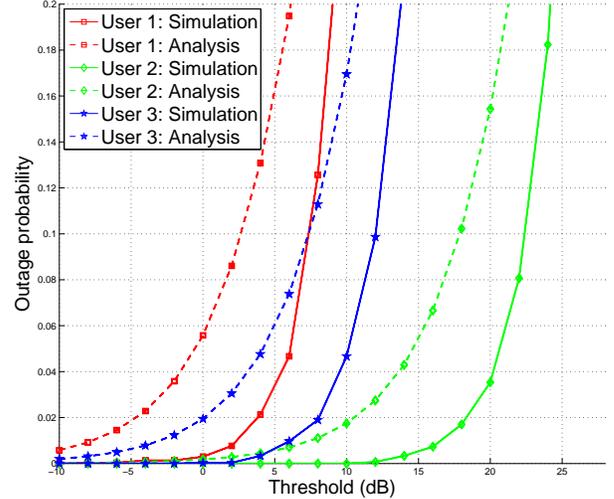}
 \caption{Outage performance of GSM with LTE overlay by simulation and analysis: The positions of the 3 GSM users are the same as those in Fig.\ref{fig:outSvsA}.}
 \label{fig:outSvsALTE}
 \end{figure}

  \begin{figure}
  \centering
  \includegraphics[width=8cm]{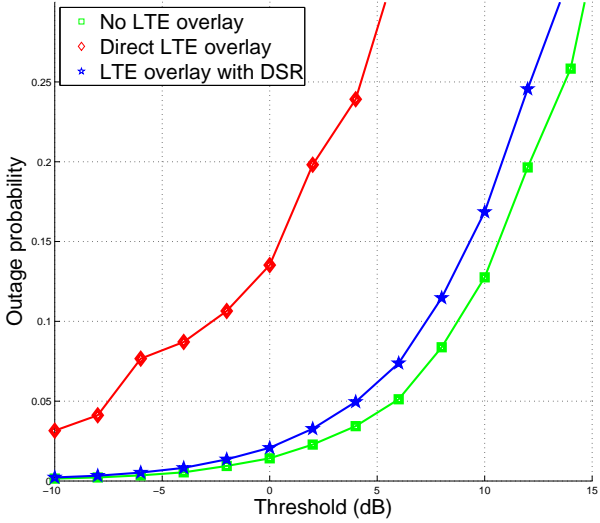}
  \caption{Comparison of the GSM outage performance in the three scenarios: No LTE overlay, Direct LTE overlay (without DSR), LTE overlay with DSR. }
  \label{fig:out_350_0}
  \end{figure}

\section {Conclusions}
\label{sec:conclusion}

In this paper we propose a novel DSR solution for deploying LTE small cells in the existing GSM networks. The system design and feasibility of DSR are carefully detailed. We then adopt a specific system model to study the throughput of LTE small cells by both analysis and simulation. The impact of the overlaid small cells on the outage performance of the GSM macro cells is also carefully investigated. Our study shows that the proposed solution can provide UEs with  high speed data services when they are covered by the small cells. Meanwhile, it allows the normal operation of the existing GSM networks.

Though the focus of this paper is about deploying LTE small cells in the existing GSM networks, 
we can also apply DSR to the overlay of LTE small cells in CDMA spectrum. For example, we can allow the LTE small cell to use the whole CDMA spectrum but reserve some PRBs, which can be used for the CDMA transmissions in the macro cell.

\appendix

\subsection{Proof of Theorem \ref{thm:1}}
\label{proof:thm1}
Note that
\begin{align}
\tau_i(r) &= \b E[ B_i \log ( 1 + \frac{1}{\eta}\sinr_i (r) ) ] \notag  \\
&=B_i \int_0^\infty  \b P(\log ( 1 + \frac{1}{\eta}\sinr_i (r) ) > t )  dt  \notag  \\
&= B_i \int_0^\infty  \b P( \sinr_i (r) > \eta (e^t - 1) )  dt, 
\end{align}
where the integrand can be derived as follows.
\begin{align}
\b P( \sinr_i (r) > T ) &= \b P ( \frac{ S(r) }{ I_i + \sigma^2 } > T )  \notag \\
&= \b P \left(    H_l   > \frac{T (I_i + \sigma^2)}{\hat{P}_l l(r)} \right)  \notag \\
&= \b E[e^{ - \frac{T (I_i + \sigma^2)}{\hat{P}_l l(r)} } ] \label{eq:proof:1} \\
&= \exp( -s^\star \sigma^2 ) \cdot  \b E[ \exp( -s^\star I_i ) ], 
\end{align}
where we use $H_l \sim Exp(1)$ in (\ref{eq:proof:1}) and $s^\star = \frac{T }{\hat{P}_l  } r^{\alpha}$.

Next we derive the Laplace transform $\c L_{I_i} (s)$.
\begin{align}
\c L_{I_i} (s)  &=\b E[ \exp( -s I_i ) ]   \notag \\
&= \b E_{ \Phi, F_x } \left[ e^{- s \sum_{ x \in \Phi \slash B(o,R_i) } \hat{P}_g F_x G( \Theta_x ) l( \| x \| )   } \right]   \notag \\
&=\b  E_{ \Phi, F_x } \left[ \prod_{x \in \Phi \slash B(o,R_i)} e^{- s \hat{P}_g F_x G( \Theta_x ) l( \| x \| )   } \right]   \notag \\
&= \b  E_{ \Phi  } \left[ \prod_{x \in \Phi \slash B(o,R_i)} \b E_{F} \left[ e^{- s \hat{P}_g F G( \Theta_x ) l( \| x \| )   } \right] \right]  \notag \\
& = e^{  -  \lambda \int_{0}^{2\pi} \int_{R_i}^\infty  v\left( 1 - \b E_{F} \left[ e^{- s \hat{P}_g F G( \phi ) l( v )   } \right] \right)    dv  d\phi    }  \label{eq:pbgf} \\
&= e^{  -  \lambda \int_{0}^{2\pi} \int_{R_i}^\infty  v\left( 1 - \c L_{F} ( s \hat{P}_g  G( \phi ) l( v ) ) \right)    dv  d\phi    }  \notag \\
&=  e^{  \lambda \pi R_i^2 +  \lambda \int_{0}^{2\pi} \int_{R_i}^\infty  v \c L_{F} ( s \hat{P}_g  G( \phi ) l( v ) )     dv  d\phi  } \notag
\end{align}
where
$
\c L_{F} (\omega) = \b E_{F} \left[ \exp(- \omega F  ) \right]
$. Denote by $Z(\phi) = \hat{P}_g  G( \phi )  F$. In (\ref{eq:pbgf}), we apply the probability generating functional of PPP (see, e.g., \cite{stoyan1995stochastic,baccelli2009stochastic}).
Substituting $l(v) = v^{-\alpha} $ yields
\begin{align}
&\int_{R_i}^\infty  v \c L_{F} ( s \hat{P}_g  G( \phi ) v^{-\alpha} )     dv  \notag \\
&= \b E_{Z} \left[ \int_{R_i}^\infty  v   \exp(- s  v^{-\alpha} Z(\phi) )      dv \right]  \notag \\
&= \frac{ s^{ \frac{2}{\alpha} } }{\alpha} \b E_{Z} \left[   Z^{ \frac{2}{\alpha} }(\phi) \cdot  \gamma ( -\frac{2}{\alpha}, sZ(\phi) R_i^{ - \alpha } ) \right], \notag 
\end{align}
where $\gamma(s,x) = \int_{0}^{x} t^{s-1} e^{-t} dt $ is the lower incomplete gamma function.

Finally, plugging $N_{ue} = \lambda^{(u)} \pi R_c^2$ and $\tau_i (r)$ yields the desired expression for $\tau (r)$.

\subsection{Proof of Theorem \ref{thm:2}}
\label{proof:thm2}

Note that $P_{out} ( d, \beta ) = 1- P_{c} ( d, \beta )$ where $P_{c} ( d, \beta )$ equals
\begin{align}
&\b  P( \frac{ S_g ( d, \beta ) }{ I_{0} + I_{l} + \sigma^2 } \geq T  ) \notag \\
&=\b  P\left(  F_l   \geq \frac{T( I_{0} + I_{l} + \sigma^2)}{\hat{P}_g G(\beta) l(d)}   \right) \notag\\
&=\b  E_{I_o,I_l}\left[ \b  P\left(  F_l   \geq \frac{T( I_{0} + I_{l} + \sigma^2)}{\hat{P}_g G(\beta) l(d)} \mid I_o,I_l  \right)    \right] \notag \\
&=\b  E_{I_o,I_l}\left[e^{ -s^\star ( I_0 + I_l + \sigma^2) } \right]  \notag \\
&= e^{-s^\star\sigma^2} \cdot \b E_{I_0} [e^{ -s^\star  I_0} ] \cdot \b  E_{I_l} [e^{ -s^\star  I_l} ], \notag
\end{align}
where $s^\star   = \frac{T}{\hat{P}_g G(\beta) l(d)}$. Note $\b E_{I_0} [e^{ -s  I_0} ]$ has been derived in the proof of Theorem \ref{thm:1}. The radius $R_0$ of the guard region follows from simple geometry analysis which yields
\begin{align}
R_0 &= \sqrt{ ( \frac{\sqrt{3}}{2} R_m - D\sin \beta )^2 + ( \frac{3}{2} R_m + D \cos \beta )^2 }.\notag
\end{align}
If $F_l \sim Exp(1)$, then
\begin{align}
\b E_{I_l} [e^{ -s  I_l} ] =  E_{F_l} [e^{ -s  \hat{P}_l  l(r) \cdot F_l} ] = \frac{1}{1+ s  \hat{P}_l r^{-\alpha}}.\notag
\end{align}
This completes the proof.

\bibliographystyle{IEEEtran}
\bibliography{IEEEabrv,Reference}

\begin{thebibliography}{10}
\providecommand{\url}[1]{#1}
\csname url@samestyle\endcsname
\providecommand{\newblock}{\relax}
\providecommand{\bibinfo}[2]{#2}
\providecommand{\BIBentrySTDinterwordspacing}{\spaceskip=0pt\relax}
\providecommand{\BIBentryALTinterwordstretchfactor}{4}
\providecommand{\BIBentryALTinterwordspacing}{\spaceskip=\fontdimen2\font plus
\BIBentryALTinterwordstretchfactor\fontdimen3\font minus
  \fontdimen4\font\relax}
\providecommand{\BIBforeignlanguage}[2]{{%
\expandafter\ifx\csname l@#1\endcsname\relax
\typeout{** WARNING: IEEEtran.bst: No hyphenation pattern has been}%
\typeout{** loaded for the language `#1'. Using the pattern for}%
\typeout{** the default language instead.}%
\else
\language=\csname l@#1\endcsname
\fi
#2}}
\providecommand{\BIBdecl}{\relax}
\BIBdecl

\bibitem{cisco2011cisco}
Cisco, ``Cisco visual networking index: Global mobile data traffic forecast
  update, 2011-2016,'' \emph{white paper}, February 2012.

\bibitem{website:3gppLTE}
\BIBentryALTinterwordspacing
3GPP, ``{3GPP LTE Official Website},'' 2012. [Online]. Available:
  \url{http://www.3gpp.org/LTE}
\BIBentrySTDinterwordspacing

\bibitem{Damnjanovic2011survey}
A.~Damnjanovic, J.~Montojo, Y.~Wei, T.~Ji, T.~Luo, M.~Vajapeyam, T.~Yoo,
  O.~Song, and D.~Malladi, ``A survey on {3GPP} heterogeneous networks,''
  \emph{IEEE Wireless Communications}, vol.~18, no.~3, pp. 10--21, June 2011.

\bibitem{Ghosh2012Heterogeneous}
A.~Ghosh, N.~Mangalvedhe, R.~Ratasuk, B.~Mondal, M.~Cudak, E.~Visotsky,
  T.~Thomas, J.~Andrews, P.~Xia, H.~Jo, H.~Dhillon, and T.~Novlan,
  ``Heterogeneous cellular networks: {From} theory to practice,'' \emph{IEEE
  Communications Magazine}, vol.~50, no.~6, pp. 54--64, June 2012.

\bibitem{mehrotra1997gsm}
A.~Mehrotra, \emph{GSM System Engineering}.\hskip 1em plus 0.5em minus
  0.4em\relax Artech House Inc., 1997.

\bibitem{lin2013dynamic}
X.~Lin and H.~Viswanathan, ``Dynamic spectrum refarming with overlay for legacy
  devices,'' \emph{submitted to IEEE Transactions on Wireless Communications},
  February 2013. Available at arXiv preprint arXiv:1302.0320.

\bibitem{3gppMobility}
3GPP, ``Evolved universal terrestrial radio access ({E-UTRA}); physical
  channels and modulation,'' \emph{3GPP TS 36.211 V10.5.0}, June 2012.

\bibitem{brown2000cellular}
T.~X. Brown, ``Cellular performance bounds via shotgun cellular systems,''
  \emph{IEEE Journal on Selected Areas in Communications}, vol.~18, no.~11, pp.
  2443--2455, 2000.

\bibitem{andrews2010tractable}
J.~G. Andrews, F.~Baccelli, and R.~Ganti, ``A tractable approach to coverage
  and rate in cellular networks,'' \emph{IEEE Transactions on Communications},
  vol.~59, no.~11, pp. 3122--3134, November 2011.

\bibitem{lin2012towards}
X.~Lin, R.~K. Ganti, P.~J. Fleming, and J.~G. Andrews, ``Towards understanding
  the fundamentals of mobility in cellular networks,'' \emph{IEEE Transactions
  on Wireless Communications}, accepted, January 2013. Available at arXiv
  preprint arXiv:1204.3447.

\bibitem{lin2013novel}
X.~Lin and H.~Viswanathan, ``A novel approach to supporting legacy devices in
  {LTE} networks,'' \emph{submitted to IEEE Globecom}, December 2013.

\bibitem{3gppMimoSim}
3GPP, ``Spatial channel model for multiple input multiple output {(MIMO)}
  simulations {(Release 11)},'' \emph{3GPP TR 25.996 V11.0.0}, September 2012.

\bibitem{stoyan1995stochastic}
D.~Stoyan, W.~Kendall, J.~Mecke, and L.~Ruschendorf, \emph{Stochastic Geometry
  and Its Applications}.\hskip 1em plus 0.5em minus 0.4em\relax Wiley New York,
  1995.

\bibitem{baccelli2009stochastic}
F.~Baccelli and B.~Blaszczyszyn, \emph{Stochastic Geometry and Wireless
  Networks - Part I: Theory}.\hskip 1em plus 0.5em minus 0.4em\relax Now
  Publishers Inc, 2009.

\end{thebibliography}

\end{document}